\documentclass[11pt]{article}
\usepackage[T1]{fontenc}
\usepackage[utf8]{inputenc}
\usepackage{lmodern}
\usepackage[letterpaper,margin=1in]{geometry}
\usepackage{microtype}
\usepackage{amsmath,amssymb,amsthm,mathtools}
\usepackage{aliascnt}
\usepackage{booktabs,array,tabularx}
\usepackage{graphicx}
\usepackage{algorithm,algpseudocode}
\usepackage{enumitem}
\usepackage{authblk}
\usepackage[backend=biber,style=alphabetic,maxbibnames=99,maxalphanames=4,giveninits=true,doi=true,url=false,isbn=false]{biblatex}
\usepackage[hidelinks]{hyperref}
\usepackage[nameinlink,noabbrev]{cleveref}
\usepackage{xcolor}
\makeatletter
\providecommand{\theHALG@line}{}
\renewcommand{\theHALG@line}{\thealgorithm.\arabic{ALG@line}}
\makeatother
\hypersetup{pdftitle={Optimal (Parallel) Spooky Pebbling on Binary Trees},pdfauthor={Mingyu Lee, Sanghyun Lee, Kabgyun Jeong},pdfsubject={Space, work, depth, and cumulative occupancy in measurement-assisted pebbling}}

\newtheorem{theorem}{Theorem}[section]
\newaliascnt{lemma}{theorem}
\newtheorem{lemma}[lemma]{Lemma}
\aliascntresetthe{lemma}
\newaliascnt{proposition}{theorem}
\newtheorem{proposition}[proposition]{Proposition}
\aliascntresetthe{proposition}
\newaliascnt{corollary}{theorem}

\aliascntresetthe{corollary}
\theoremstyle{definition}
\newaliascnt{definition}{theorem}

\aliascntresetthe{definition}
\theoremstyle{remark}
\newaliascnt{remark}{theorem}

\aliascntresetthe{remark}
\crefname{lemma}{lemma}{lemmas}
\crefname{proposition}{proposition}{propositions}
\crefname{corollary}{corollary}{corollaries}
\crefname{definition}{definition}{definitions}
\crefname{remark}{remark}{remarks}
\DeclareMathOperator{\pred}{pred}
\DeclareMathOperator{\parv}{par}

\DeclareMathOperator{\popcount}{popcount}

\newcommand{\Peb}{\mathsf{Pebble}}
\newcommand{\ket}[1]{\lvert#1\rangle}
\newcommand{\Unp}{\mathsf{Unpebble}}
\newcommand{\Ghost}{\mathsf{Ghost}}
\newcommand{\Fast}{\mathsf{SequentialPebble}}
\newcommand{\Release}{\mathsf{Release}}

\newcommand{\Sweep}{\mathsf{ParallelPebble}}
\newcommand{\Clean}{\mathsf{FastClean}}
\newcommand{\Resolve}{\mathsf{Resolve}}
\newcommand{\Contract}{\mathsf{Contract}}

\newcommand{\Area}{\mathcal A}
\newcommand{\Wseq}{W_{\mathrm{seq}}}
\newcommand{\Wpar}{W_{\mathrm{par}}}
\newcommand{\Dpar}{D_{\mathrm{par}}}

\setlist{itemsep=2pt,topsep=4pt}
\allowdisplaybreaks[1]
\floatname{algorithm}{Algorithm}
\algrenewcommand\algorithmicrequire{\textbf{Input:}}
\algrenewcommand\algorithmicensure{\textbf{Output:}}

\title{Optimal (Parallel) Spooky Pebbling on Binary Trees}

\author[1,2]{Mingyu Lee}
\author[3,2]{Sanghyun Lee}
\author[4,2]{Kabgyun Jeong}

\affil[1]{Department of Computer Science and Engineering,
Seoul National University, Seoul 08826, Republic of Korea}

\affil[2]{Team QST,
Seoul National University, Seoul 08826, Republic of Korea}

\affil[3]{Department of Mechanical Engineering,
Kyung Hee University, Yongin 17104, Republic of Korea}

\affil[4]{Institute of Computer Technology, College of Engineering,
Seoul National University, Seoul 08826, Republic of Korea}

\date{}

\begin{document}
\maketitle
\vspace{-1.4em}

\begin{abstract}

Pebble games model computations under a fixed space budget. Spooky pebbling allows quantum memory to be released by measurement, with the resulting phases corrected later. We study two-input computations with binary-tree dependencies and determine the optimal work and parallel depth for complete binary trees. Our key idea is to clean up the tree in blocks, reducing repeated recomputation of intermediate values. For the complete tree $B_h$ with $n=2^h-1$ vertices and every space budget $h+1\le s\le n$, we give an algorithm with asymptotically optimal work $\Theta(nh/\log(s+1))$. At the minimum budget $s=h+1$, this improves the $O(n\log n)$ bound of Kornerup, Sadun, and Soloveichik to $\Theta(n\log n/\log\log n)$, resolving their time-optimality question.

We also construct a parallel schedule with optimal depth
\[
 \Theta\!\left(h+\frac ns
 \max\!\left\{\frac{h}{\log(s+1)},\,1+\log^*h\right\}\right).
\]
Our lower bounds hold for every full binary tree. We also show that achieving optimal parallel depth can require asymptotically more work than minimizing work alone. Applying our schedule to the RNS point-addition trees in the public implementation of Chevignard, Fouque, and Schrottenloher reduces their Toffoli/AND gate count by $20.86\%$ for a P-224 instance and $22.66\%$ for a P-256 instance, using the same arithmetic circuits and peak workspace.

\end{abstract}

\section{Introduction}

Consider a quantum out-of-place operation $f$ with $k$ inputs:
\begin{equation}
    \ket{a_1}\cdots\ket{a_k}\ket{0} \longmapsto \ket{a_1}\cdots\ket{a_k}\ket{f(a_1,\ldots,a_k)}.
\end{equation}

This operation must act coherently on superpositions of inputs. When the output of one operation serves as an input to another, the intermediate value must be stored in a quantum register until it is used. Repeatedly combining pairs of values forms a binary tree of computations when each value is used as an input at most once. Such a structure appears when adding elliptic-curve points~\cite{CFS2026}, as we discuss in \cref{sec:cfs-case-study}.

Once an intermediate value is no longer needed, we would like to clear its register to save space. We can do this by reversing the operation that produced the value, but this requires its inputs to remain available. If their registers have already been cleared, those inputs must be computed again. Keeping intermediate values therefore uses more memory, while clearing them early may require more computation.

Pebble games describe this tradeoff between memory and computation~\cite{Bennett1989}. Each vertex represents a value, and its predecessors represent the inputs needed to compute it. A pebble indicates that the value is stored in a register. Placing a pebble corresponds to computing the value, and removing it with its inputs available corresponds to uncomputing it.

Measurement offers another way to release a quantum register. Instead of uncomputing its value, we can measure the register in the Hadamard basis and free it immediately. This reduces the quantum memory in use, but leaves a phase that must be corrected later. Suppose the measured register holds $x$, and let $\ket{\psi_x}$ denote the state of the remaining registers. A measurement outcome $m$ leaves
\begin{equation}
    \sum_x \alpha_x\ket{\psi_x}\ket{x} \longmapsto \sum_x \alpha_x(-1)^{m\cdot x}\ket{\psi_x},
\end{equation}
where $m\cdot x$ is computed modulo two and the measured register has been discarded. The factor $(-1)^{m\cdot x}$ may change the relative phases of the remaining state and therefore affect interference in subsequent quantum operations.

In the pebble game, we represent this pending phase correction by a ghost. Unlike a pebble, a ghost does not occupy a quantum register. The measurement outcome is stored classically. Correcting the phase is called exorcising the ghost. Once the required inputs are available, we can do this by recomputing $x$ and applying $Z$ to each qubit whose corresponding measurement outcome is one. We obtain
\begin{equation}
    \sum_x \alpha_x(-1)^{m\cdot x}\ket{\psi_x}\ket{x} \xrightarrow{\bigotimes_j Z^{m_j}} \sum_x \alpha_x\ket{\psi_x}\ket{x}.
\end{equation}
We can then uncompute $x$ to clear its register without leaving a phase.

This extension of the pebble game, introduced by Gidney, is called spooky pebbling~\cite{Gidney2019}. It allows us to release quantum memory before the inputs needed for uncomputation are available, at the cost of correcting the resulting phases later. A successful strategy must produce the desired output and clear all other pebbles and ghosts. Figure~\ref{fig:binary-pebbling-concept} illustrates a configuration of the spooky pebble game on a binary tree.

\begin{figure}[htbp]
    \centering
    \includegraphics[width=0.8\textwidth]{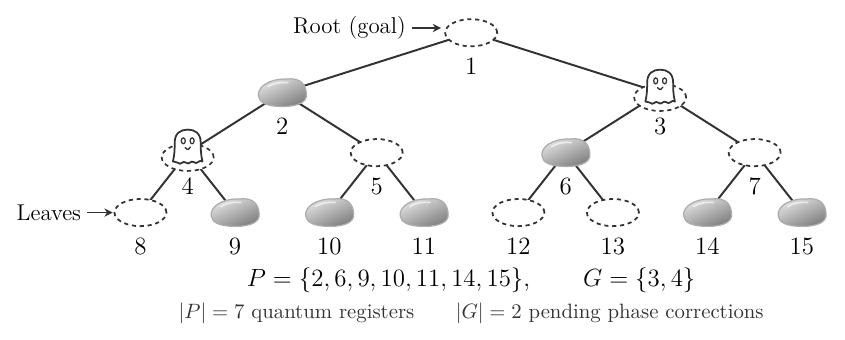}
    \caption{Illustration of spooky pebbling on a binary tree. Pebbles represent values stored in quantum registers, and ghosts represent pending phase corrections. Vertices with neither a pebble nor a ghost are empty. The configuration shown here occupies 7 quantum registers. Pebbling or unpebbling an internal vertex requires pebbles on both children. The goal is to leave only the root pebbled, with all other vertices empty.}
    \label{fig:binary-pebbling-concept}
\end{figure}

\subsection{Prior work}

Kornerup, Sadun, and Soloveichik studied sequential spooky pebbling, where one pebble is placed or removed at a time~\cite{KSS2025}. They determined the optimal tradeoff between computation time and memory use for a chain of connected nodes. For a complete binary tree with $h$ vertex levels and $n=2^h-1$ vertices, they constructed a strategy using the minimum $h+1$ pebbles and $O(n\log n)$ steps~\cite[Theorem~5.5]{KSS2025}. Their step count includes placements, uncomputations, and ghosting. The strategy first computes the root, leaving ghosts below it, and then recomputes child values to clear these ghosts. It remained open whether the $O(n\log n)$ bound can be improved, along with the question of a general time--space tradeoff on binary trees~\cite[Section~6]{KSS2025}.

On more general dependency graphs, KSS studied the hardness of minimizing pebble space. Quist and Laarman proved PSPACE-completeness results and developed SAT-based methods to find spooky pebbling strategies with small quantum space, classical space, or execution time~\cite{QL2025}. Our focus is on explicit strategies and matching bounds for binary trees.

Another way to reduce execution time is to perform independent operations in parallel. To formalize this in reversible computing, Blocki, Holman, and Lee introduced parallel reversible pebbling~\cite{BHL2022}. Combining this framework with spooky pebbling, Kahanamoku-Meyer, Ragavan, and Van Kirk (KMRV) introduced parallel spooky pebbling, which also allows registers to be released by measurement~\cite{KMRV2026}. In their model, operations are executed in synchronous rounds: the number of rounds determines the depth, while the total number of computations and uncomputations determines the work. Performing multiple operations in a single round can therefore reduce depth without increasing total work.

KMRV studied a linear chain with a fixed initial pebble and $\ell$
additional vertices. They showed that visiting the terminal vertex
and then clearing these additional vertices takes an optimal $2\ell$
rounds using $O(\log\ell)$ pebbles~\cite{KMRV2026}. They also applied
their construction to reduce the depth of arithmetic in Regev's
factoring algorithm.

These results leave the tradeoffs on binary trees open. Independent subtrees can be computed in parallel, but both child values must be available to compute or uncompute their parent. Under a fixed space budget, a schedule must decide which subtrees to process concurrently and which values to retain for cleanup. We ask what work and depth are optimal at each space budget, and whether a single strategy can minimize both.

\subsection{Main results}

Let $B_h$ be the complete binary tree of height $h\ge2$, where height counts vertices on a longest root-to-leaf path. Then $B_h$ has $n=2^h-1$ vertices. Write $\Wseq^\star(B_h,s)$ and $\Wpar^\star(B_h,s)$ for the minimum sequential and parallel work with space at most $s$, and $\Dpar^\star(B_h,s)$ for the minimum parallel depth. Also write $\lambda(x)=1+\log^*x$, where $\log^*x$ is the number of iterated binary logarithms needed to reach a value at most one.

\begin{theorem}[Informal, optimal space--work tradeoffs and parallel depth,
see \cref{thm:seq-upper,thm:parallel-upper,thm:universal-work,thm:weighted}]
\label{thm:main}
For every integer $h+1\le s\le n$,
\begin{equation}
    \Wseq^\star(B_h,s)=\Wpar^\star(B_h,s)=\Theta\!\left(\frac{nh}{\log(s+1)}\right).
    \label{eq:main-work}
\end{equation}
The optimal parallel depth is
\begin{equation}
    \Dpar^\star(B_h,s)=\Theta\!\left(h+\frac ns\max\!\left\{\frac{h}{\log(s+1)},\lambda(h)\right\}\right).
    \label{eq:main-depth}
\end{equation}
Work and depth are minimized separately. All implicit constants are absolute and uniform over the entire feasible interval.
\end{theorem}

At the minimum budget $s=h+1$, a single parallel construction attains
\begin{equation}
    W=\Theta\!\left(\frac{n\log n}{\log\log n}\right), \qquad D=\Theta\!\left(\frac n{\log\log n}\right).
    \label{eq:min-space}
\end{equation}
The work bound improves the KSS bound by a factor of $\log\log n$ and resolves their question about time-optimality on complete binary trees. At space $s=\Theta(n)$, linear work and depth $\Theta(h)$ are attainable.

Our algorithms reduce repeated reconstruction during cleanup through
\emph{block cleanup}. We retain the values along a block's boundary while
computing and uncomputing the intermediate vertices above it. For parallel
execution, we also overlap independent boundary computations and the
cleanup of disjoint subtrees.

For an associative operation, we can change the parenthesization
without changing the input order or the result.
Each parenthesization gives an evaluation tree.
We therefore ask which tree shape minimizes work and parallel
depth at a common space budget.

A full binary tree has exactly two children at each internal vertex.
Our lower bounds hold for every such tree, including unbalanced trees,
by accounting for the recomputation and storage required during cleanup.
Our complete-tree constructions attain these bounds, yielding the
following optimality result at a common space budget.

\begin{theorem}[Optimality over tree shapes]
\label{thm:main-shape}
Let $N=2^{h-1}$, and let $\mathcal T_N$ contain all ordered full binary trees with $N$ leaves. At every common budget $h+1\le s\le2N-1$, a complete balanced tree minimizes sequential work, parallel work, and parallel depth up to constant factors over $\mathcal T_N$.
\end{theorem}

Work and depth can require different strategies. At $s=\Theta(n/h)$, the minimum work is $\Theta(n)$, while the minimum depth is $\Theta(h\lambda(h))$. Every strategy attaining this depth performs $\Omega(n\lambda(h))$ work, and our construction attains that work bound. Thus a strategy that minimizes depth in this regime must do asymptotically more work than one that minimizes work alone.

We also apply our cleanup strategy to a fixed point-addition tree in the public implementation of Chevignard, Fouque, and Schrottenloher~\cite{CFS2026}. Their computation uses a residue number system (RNS), evaluating the point-addition formulas modulo several small primes. With the same tree and arithmetic circuits, our strategy reduces per-prime Toffoli/AND gate count by $20.86\%$ for a P-224 instance and $22.66\%$ for the primary P-256 instance, with unchanged peak workspace. These measurements use the costs of the individual circuits rather than unit costs for tree operations. \Cref{sec:cfs-case-study} gives the comparison.

\section{Pebbling models}
\label{sec:models}

\subsection{Fixed dependency trees}

A rooted full binary tree $T=(V,E)$ has either zero or two children at each vertex. Edges point from children to their parent, so the root $r$ is the unique sink. We write $\pred(v)$ for the children of $v$, and $\parv(v)$ for its parent, when it exists. The height $H(T)$ is the number of vertices on a longest leaf-to-root path. We fix the dependency tree $T$. The complete binary tree $B_h$ has
$2^h-1$ vertices and $2^{h-1}$ leaves.
All logarithms are binary unless stated otherwise.

An internal vertex $v$ computes a function $f_v$ from its two child values using an out-of-place unitary $U_v$ satisfying
\[
U_v\ket{a}\ket{b}\ket{0}=\ket{a}\ket{b}\ket{f_v(a,b)}.
\]
Applying $U_v^\dagger$ uncomputes this value when both child values are available. Leaf values have no dependencies and can be loaded and erased at any time during the computation. The functions $f_v$ may differ between vertices and need not be associative or invertible.
When every internal vertex applies the same associative operation,
we may choose the evaluation tree before the computation while
preserving the leaf order. Each strategy then follows that fixed tree.

\subsection{Configurations and actions}

A configuration is a pair $(P,G)$ of disjoint subsets of $V$. Vertices in $P$ have pebbles, representing values stored in quantum registers. Vertices in $G$ have ghosts, representing pending phase corrections. All other vertices are empty. We allow three actions:
\[
\begin{array}{lllc}
\text{Action} & \text{Condition} & \text{New configuration} & \text{Work} \\[3pt]
\Peb(v) & v\notin P,\ \pred(v)\subseteq P & (P\cup\{v\},\ G\setminus\{v\}) & 1 \\[3pt]
\Unp(v) & \{v\}\cup\pred(v)\subseteq P & (P\setminus\{v\},\ G) & 1 \\[3pt]
\Ghost(v) & v\in P & (P\setminus\{v\},\ G\cup\{v\}) & 0
\end{array}
\]
When $v\in G$, $\Peb(v)$ also applies the single-qubit phase corrections specified by the stored measurement outcome, clearing the ghost.

The initial configuration is $(\varnothing,\varnothing)$, and the goal is $(\{r\},\varnothing)$: retain the root value and clear every other pebble and every ghost.

\subsection{Execution and cost metrics}
\label{sec:work-counts}

A sequential strategy executes one action at a time. Each $\Peb$, $\Unp$, or $\Ghost$ action is one step.

A parallel strategy executes actions in rounds with two phases, following the execution rule of KMRV~\cite[Section~2.1]{KMRV2026}. In the first phase, $\Peb$ and $\Unp$ actions may run simultaneously if their active sets
\begin{equation}
A_v=\{v\}\cup\pred(v)
\label{eq:active-set}
\end{equation}
are pairwise disjoint. Their preconditions are checked at the beginning of that phase. In the second phase, $\Ghost$ may be applied to any set of currently pebbled vertices.

Parallel depth is the number $D$ of rounds. When operations on disjoint logical registers can run concurrently without additional hardware constraints, parallel depth models execution time more directly than the sequential step count.

Each pebble occupies one register, which may contain multiple qubits. We assume equal register widths and measure space by counting pebbles, including those on leaves. Ghosts occupy no quantum space. Sequential space is the maximum number of pebbles present at any step. For a parallel strategy $\mathcal P$, let $P_t$ be the pebble set after round $t$, and let $P_t^+$ be its pebble set after the first phase, before ghosting. We count every register occupied during a round:
\begin{equation}
U_t=P_{t-1}\cup P_t^+, \qquad
s(\mathcal P)=\max_{1\le t\le D}|U_t|, \qquad
\Area(\mathcal P)=\sum_{t=1}^D|U_t|.
\label{eq:parallel-resources}
\end{equation}
Here $\Area$ is cumulative storage. External input registers and temporary registers used inside an action are accounted for separately in circuit resource estimates. Classical storage and control are unrestricted.

Work counts function computations and uncomputations:
$W=\#\Peb+\#\Unp$.
$\Ghost$ releases a register by measurement without evaluating a vertex function, so it contributes no work.

Let $\Pi$ denote the total number of pebble placements in a complete strategy. Every placed pebble except the final root must eventually be removed by $\Unp$ or $\Ghost$, so there are $\Pi-1$ removals. The total action count is therefore $N_{\mathrm{step}}=2\Pi-1$, giving
\begin{equation}
W\le N_{\mathrm{step}}<2W.
\end{equation}
KSS~\cite{KSS2025} count ghosting as a step, whereas KMRV~\cite{KMRV2026} count only computations and uncomputations. The two conventions differ by at most a factor of two, so our asymptotic work bounds are unchanged.

For a dependency tree $T$ and an integer space budget $s$, let $\Wseq^\star(T,s)$ and $\Wpar^\star(T,s)$ denote the minimum sequential and parallel work, and let $\Dpar^\star(T,s)$ denote the minimum parallel depth. If no valid strategy exists within space $s$, the corresponding value is $+\infty$.

For complete binary trees $B_h$ with $h\ge2$, computing the root requires at least $h+1$ pebbles even when measurement-assisted register release is permitted~\cite{Cook1973,KSS2025}. A budget of $n=2^h-1$ pebbles allows the entire tree to be held in memory. We therefore study the feasible interval
\begin{equation}
h+1\le s\le2^h-1.
\end{equation}

Any parallel strategy can be serialized by executing each round's actions in order: first $\Unp$, then $\Peb$, and finally $\Ghost$. Pairwise disjoint active sets preserve the preconditions, and every intermediate pebble configuration stays within $U_t$. Conversely, each sequential step can be executed as a single-action parallel round. Both conversions preserve work and respect the same space budget, so
\begin{equation}
\Wseq^\star(T,s)=\Wpar^\star(T,s).
\label{eq:serialize}
\end{equation}

\section{Our algorithms}
\label{sec:main-algorithms}

The computation has two parts: compute the root value from the bottom, then clear every other vertex while retaining the root. With space $s=n$, we can hold all vertex values at once. Pebbling the tree bottom-up and unpebbling every non-root vertex top-down uses $n+(n-1)=2n-1$ work. Under a smaller space budget, values must be released before the cleanup that needs them, and recovering those values requires additional work.

The bottom-up \texttt{fast\_pebble} routine of KSS, which we call $\Fast$, computes the root with $O(n)$ work using the minimum feasible $h+1$ pebbles, but leaves ghosts below it. KSS gave an $O(n\log n)$-step strategy to clear these ghosts at the same minimum budget~\cite[Theorem~5.5]{KSS2025}. The additional work comes from reconstructing child values during cleanup. We reduce this recomputation through block cleanup, introduced in
\cref{sec:block-cleanup}. For parallel execution, we also overlap
independent computations and cleanup jobs to reduce depth.
Dividing the space budget among jobs can require additional
recomputation, leading to the work--depth tradeoff examined
in \cref{sec:separation}.

\subsection{Bottom-up root calculation}
\label{sec:fast-root}

For a complete subtree $T$ rooted at $v$ with $H(T)=r$, $\Fast(T)$ computes its children recursively, pebbles $v$, and releases the children afterward. An internal child is ghosted because its own children have already been released. For a pebbled vertex $u$, let $\Release(u)$ mean $\Unp(u)$ if $u$ is a leaf, and $\Ghost(u)$ otherwise.

\begin{algorithm}[H]
\caption{Bottom-up root computation $\Fast(T)$}
\label{alg:fast-root}
\begin{algorithmic}[1]
\Require A complete subtree $T$ rooted at $v$, with every vertex initially empty or ghosted.
\Ensure The root $v$ is pebbled. Every other internal vertex of $T$ is ghosted, and every other leaf is empty.
\If{$T$ consists of one vertex}
 \State $\Peb(v)$; \Return
\EndIf
\State Run $\Fast(T_L)$ and $\Fast(T_R)$ on the child subtrees
\State $\Peb(v)$
\For{each child root $u$ of $v$}
\State $\Release(u)$
\EndFor
\end{algorithmic}
\end{algorithm}

Every vertex is pebbled once, and every leaf other than the root is unpebbled once. Thus $\Fast(T)$ uses $\Theta(2^r)$ work. From an empty tree, this is optimal up to constants for root computation because every vertex must be computed at least once. Space usage peaks at $r+1$ pebbles for $r\ge2$, or one pebble for a leaf.

\subsection{Block cleanup}
\label{sec:block-cleanup}

Our key idea is to clean up the tree in \emph{blocks}. Starting at a subtree's root, we extend a block downward as far as the space budget permits. Its bottom vertices form the boundary $F$. Holding their values lets us compute and uncompute the vertices above $F$ without repeatedly reconstructing intermediate values. After releasing $F$, we recursively clean the subtrees rooted there.

For a complete subtree of height $r$, \cref{alg:boundary} constructs a balanced boundary of $m$ vertices, where $2\le m\le2^{r-1}$. Its vertices lie at depths $\lfloor\log m\rfloor$ or $\lceil\log m\rceil$. We prepare their values with $\Fast$, processing the tallest boundary subtrees first.

\begin{algorithm}[H]
\caption{Constructing a balanced boundary $\mathsf{Boundary}(T,m)$}
\label{alg:boundary}
\begin{algorithmic}[1]
\Require A complete subtree $T$ rooted at $v$ with $H(T)=r$, and a desired boundary size $m$ satisfying $1\le m\le2^{r-1}$.
\Ensure A set $F$ of exactly $m$ vertices forming a balanced boundary of $T$.
\If{$m=1$}
 \State \Return $\{v\}$
\EndIf
\State $F_L\gets\mathsf{Boundary}(T_L,\lceil m/2\rceil)$
\State $F_R\gets\mathsf{Boundary}(T_R,\lfloor m/2\rfloor)$
\State \Return $F_L\cup F_R$
\end{algorithmic}
\end{algorithm}

The one-block operation $\Contract(T,m,\rho)$ clears the vertices above $F$, retaining the root when $\rho=1$, and returns $F$ for recursive cleanup. Its cost determines how large a boundary the budget permits.

\begin{algorithm}[H]
\caption{One-block cleanup $\Contract(T,m,\rho)$}
\label{alg:block-contract}
\begin{algorithmic}[1]
\Require A complete subtree $T$ with $H(T)=r\ge2$. Its leaves are empty, and each internal vertex is either empty or ghosted.
\Require A boundary size $m$ satisfying $2\le m\le2^{r-1}$, and a flag $\rho\in\{0,1\}$ indicating whether to retain the root of $T$.
\Ensure The boundary $F$ is returned, and all ghosts in the block above $F$ are exorcised. Within $T$, only its root remains pebbled when $\rho=1$, and no pebbles remain there when $\rho=0$. Any remaining ghosts lie in the subtrees rooted at $F$.
\State $F\gets\mathsf{Boundary}(T,m)$
\State Compute and retain $F$ using $\Fast$, tallest subtrees first
\State Pebble the vertices above $F$ bottom-up
\State Unpebble them top-down, except the root when $\rho=1$
\State Apply $\Release(u)$ to each $u\in F$
\State \Return $F$
\end{algorithmic}
\end{algorithm}

\begin{lemma}[Cost of one-block cleanup]
\label{lem:contract}
For $r\ge2$ and $2\le m\le2^{r-1}$, $\Contract(T,m,0)$ has peak space and work
\begin{align}
 S(r,m)&=\max\{2m-1,\ r+m-\lceil\log m\rceil\}, \label{eq:boundary-space}\\
 K(r,m)&=3\cdot2^{r-1}+m-2. \label{eq:block-work}
\end{align}

$\Contract(T,m,1)$ retains the root, reducing the work by one without changing the peak space.
\end{lemma}

\begin{proof}
The boundary $F$ has $m$ vertices, and the block above it has $m-1$ vertices. At the end of the upward pebbling pass, all $2m-1$ vertices are pebbled.

The other possible space peak occurs while preparing $F$. Put $q=r-\lceil\log m\rceil$. Its boundary subtrees have height $q$ or $q+1$, and we compute the taller ones first. A balanced boundary has at least two height-$q$ subtrees. Thus a taller computation starts with at most $m-3$ roots retained and uses at most $m+q-1$ pebbles. The final height-$q$ computation starts with $m-1$ roots retained and uses $m+q$ pebbles if $q\ge2$, or $m$ if $q=1$. Together with the top-block peak, this gives \cref{eq:boundary-space}.

For work, each of the $2^r-1$ vertices is pebbled once. The $2^{r-1}$ leaves are unpebbled once, as are the $m-1$ vertices above $F$. All other removals are ghostings. Since work counts only pebbling and unpebbling,
\[
 K(r,m)=(2^r-1)+2^{r-1}+(m-1).
\]
The top vertices are unpebbled while their children are present. Keeping the output root skips one unpebbling.
\end{proof}

For an equal-depth boundary with $b$ registers allocated to the current block, the space conditions become
\begin{equation}
 m+q\le b,\qquad 2m-1\le b,\qquad m=2^k,\quad q=r-k.
 \label{eq:feasible-cut}
\end{equation}

\subsection{Sequential work analysis}
\label{sec:sequential}

We now combine one-block cleanup into a recursive algorithm and bound its total work. The matching lower bound is proved in \cref{sec:lower-bounds}.

At a height-$r$ cleanup call with budget $b$, choose the largest $m$ satisfying $S(r,m)\le b$. Since $S(r,m)$ is nondecreasing in $m$, we find this value by binary search. The following proposition justifies this rule within the balanced recursive scheme using $\Fast$.

\begin{proposition}[Boundary choice for $\Clean$]
\label{prop:greedy-boundary}
For complete trees with unit action costs, the largest feasible boundary minimizes work among the balanced recursive contractions defined in \cref{sec:block-cleanup}, using $\Fast$ at a fixed cleanup budget $b$. The same rule applies to the output-producing block with budget $s$ and cleanup budget $s-1$.
\end{proposition}
\begin{proof}
Let $C_r$ be the minimum cleanup work in this family at budget $b$,
starting with internal ghosts and empty leaves, so $C_1=0$.
Restrict a strategy to either child subtree, still within budget $b$.
It either performs block cleanup there directly or first computes and
releases that child's root.
In both cases the child contributes at least $C_{r-1}$.
The parent is also pebbled and unpebbled, giving
\[
 C_r\ge2C_{r-1}+2.
\]

Expanding a shallowest boundary vertex of height $q\ge2$ into its
children adds one unit to the current block's work by \cref{eq:block-work},
but replaces the later cost $C_q$ by $2C_{q-1}$.
The total change is
\[
1+2C_{q-1}-C_q\le-1.
\]
Since \cref{alg:boundary} obtains successive sizes by these expansions
and $S(r,m)$ is nondecreasing, every expansion up to the largest feasible
boundary reduces total work. For the output block, use the descendant
budget $b=s-1$ in $C_q$. Retaining the root saves one unpebbling for every
boundary choice and leaves the comparison unchanged.
\end{proof}

The procedure $\Clean$ combines one-block cleanup with recursive calls on the boundary subtrees. Its initial call is $\Clean(B_h,s,1)$, with $\rho=1$ to retain the root.

\begin{algorithm}[H]
\caption{Complete-tree block cleanup $\Clean(T,b,\rho)$}
\label{alg:seq-clean}
\begin{algorithmic}[1]
\Require A complete subtree $T$ with $H(T)=r$, no pebbles, and no ghosts on leaves.
\Require A space budget $b\ge r+1$ for $r\ge2$ ($b\ge1$ for a leaf), and a flag $\rho\in\{0,1\}$ indicating whether to retain the root of $T$.
\Ensure Exorcise $T$ and clear all its pebbles, retaining the root when $\rho=1$.
\If{$r=1$}
 \State Pebble the root if $\rho=1$; \Return
\EndIf
\State Choose the largest $2\le m\le2^{r-1}$ with $S(r,m)\le b$
\State $F\gets\Contract(T,m,\rho)$
\State For each $v\in F$, run $\Clean(T_v,b-\rho,0)$ sequentially
\end{algorithmic}
\end{algorithm}

Each recursive call starts with internal ghosts and empty leaves. Since each boundary subtree has height at most $r-1$, the remaining budget $b-\rho$ is feasible. For analysis, we write $b=s-1$ for the budget available below the output.

Each block costs constant work per leaf, so we bound how many blocks contain a given leaf. Because a balanced boundary may span two adjacent depths, we introduce $k_b(r)$, the largest feasible equal-depth cut depth, to bound the height decrease along every branch.

\begin{lemma}[Stage bound for block cleanup]
\label{lem:slack}
Suppose $b\ge h$ and $r_0\le h-1$. At each height $r$ with $2^r-1>b$, let $k_b(r)$ be the largest cut depth satisfying \cref{eq:feasible-cut}. Any sequence that decreases $r$ by at least $k_b(r)$ at each stage reaches $2^r-1\le b$ in $O(1+h/\log(b+1))$ stages.
\end{lemma}
\begin{proof}
Put $a=b-r\ge1$. This gap increases as the remaining height decreases. The two cut conditions in \cref{eq:feasible-cut} become
\[
2^k-k\le a,\qquad 2^{k+1}-1\le b.
\]
Choose $k\ge1$ maximally with $2^k\le\min\{a+1,(b+1)/2\}$. Since the subtree does not fit in $b$, we have $k<r$, and both cut conditions hold. Hence, for an absolute constant $c>0$,
\[
 k_b(r)\ge c\min\{\log(a+1),\log(b+1)\}.
\]
Each cut reduces $r$, and thus increases $a$, by at least $k_b(r)$.

Once $a\ge\sqrt b$, every nonterminal cut removes $\Omega(\log(b+1))$ levels. The remaining height is at most $h$, so there are $O(h/\log(b+1))$ such cuts. If $b\ge2h$, the initial gap already exceeds $\sqrt b$. Otherwise $h\le b<2h$. Since every cut removes at least one level, at most $\lceil\sqrt b\rceil=O(h/\log(b+1))$ cuts are needed to reach this threshold.
\end{proof}

\begin{theorem}[Work analysis of $\Clean$]
\label{thm:seq-upper}
For $n=2^h-1$ and $h+1\le s\le n$, \cref{alg:seq-clean} uses at most $s$ pebbles and work $W=O(nh/\log(s+1))$.
\end{theorem}
\begin{proof}

After the first call to $\Contract$, which costs $O(n)$ work, the remaining task is to clean boundary subtrees of height at most $h-1$ using budget $b=s-1$. For any subsequent height-$r$ subtree with $2^r-1>b$, the largest feasible balanced boundary contains at least $2^{k_b(r)}$ vertices, so every recursive branch loses at least $k_b(r)$ levels.

We count the remaining work by distributing each $\Contract$ call's cost equally among the leaves of its subtree:
\begin{itemize}
    \item By \cref{eq:block-work}, each leaf is charged $O(1)$ work per call.
    \item By \cref{lem:slack}, each leaf participates in $O(1+h/\log(b+1))$ calls, including at most one final call once its remaining subtree fits within $b$.
    \item Thus each leaf is charged $O(1+h/\log(b+1))$ work in total.
\end{itemize}
There are $(n+1)/2$ leaves. Adding the first call's work gives
\[
W=O(n)+O\!\left(n\left(1+\frac{h}{\log(b+1)}\right)\right)
 =O\!\left(\frac{nh}{\log(s+1)}\right),
\]
where $b+1=s$, $\log s=\Theta(\log(s+1))$, and $\log(s+1)\le h$ because $s\le n=2^h-1$.

The space bound follows from \cref{eq:boundary-space}.

\end{proof}

\subsection{Parallel depth upper bound}
\label{sec:parallel-upper}

In this section, we construct and analyze a parallel spooky pebbling algorithm for complete binary trees at every feasible space budget. We parallelize boundary preparation and the cleanup of independent subtrees. Equal-depth boundaries give the resulting subtrees a common height and schedule. Once block cleanup has made them small enough, we divide the space among jobs and process them in parallel batches.
At the minimum budget $s=h+1$, our schedule achieves depth $O(n/\log\log n)$. The matching lower bound in \cref{eq:universal-depth} establishes asymptotic optimality.

For the initial root computation and boundary preparation, we use $\Sweep$, a parallel counterpart of $\Fast$.

\begin{lemma}[Depth analysis of $\Sweep$]
\label{lem:sweep}
For a complete subtree $T$ of height $q$ and space budget $b\ge q+1$, a procedure $\Sweep(q,b)$ computes and retains its root in depth
\begin{equation}
O\!\left(q+\frac{2^q}{b}\right).
\label{eq:sweep}
\end{equation}

Starting from $P\cap V(T)=\varnothing$ and arbitrary $G$, each vertex is pebbled once. Every non-root vertex is ghosted after use.

\end{lemma}
\begin{proof}
A leaf needs one round and one register. If $2^q-1\le b$, compute
levelwise and ghost each level after use. If $b<16$, follow the
depth-first order of $\Fast$, replacing removals by ghosting.
This uses at most $q+1$ registers and $O(2^q)=O(2^q/b)$ rounds.
Hence assume $16\le b<2^q-1$.

For $b\ge2q$, choose the largest power of two $p\le b/8$ and put
$p=2^{\ell-1}$, $d=q-\ell$, and $R=2^d$.
Compute the $R$ bottom height-$\ell$ subtrees levelwise, starting
one every second advancement round and ghosting its previous level.
After $a\ge1$ advances, a job's old and new levels use $3p/2^a$
registers during the round. The unfinished jobs therefore occupy at most
\[
p+\frac{3p}{4}+\frac{3p}{16}+\cdots\le2p, \qquad \frac{3p}{2}+\frac{3p}{8}+\cdots=2p.
\]
Combine adjacent equal-height roots by pebbling their parent and
ghosting both children, as in a binary counter's carries. Pause the
unfinished jobs during these rounds. The retained roots represent the
power-of-two blocks of the completed prefix. Before the last completion,
there are $\popcount(j)\le d$ such roots after $j$ completions.
Including an incoming root and temporary parent, the total space is at most
\[
2p+d+2\le b/4+q+2\le b.
\]
There are $2R+O(\ell)$ advancement rounds and $R-1$ combinations.
Since $p=\Theta(b)$, the depth is $O(q+2^q/b)$.

For $q+1\le b<2q$, put $\ell=\lfloor q/4\rfloor$,
$d=q-\ell$, and $R=2^d$. Process the bottom subtrees one at a time,
using the same binary carries. Before job $j=0,\ldots,R-1$, the retained
stack has $\popcount(j)$ roots, leaving local budget
\[
b_j=b-\popcount(j)\ge b-d\ge\ell+1.
\]
Use the preceding construction, including its base cases, when
$b_j\ge2\ell$, and the ghosting depth-first order otherwise.
The combinations fit in $d+2\le q+1$ registers. If $H$ jobs require depth-first order, then
$\popcount(j)>b-2\ell\ge2d/3$, so
\[
H\le2^{-2d/3}\sum_{i=0}^d\binom di2^i=2^{\eta d}, \qquad \eta=\log3-\frac23<1.
\]
The other jobs have $b_j=\Omega(b)$ and take $O(\ell+2^\ell/b)$
rounds each, while a depth-first job takes $O(2^\ell)$ rounds.
Summing these depths and the root combinations gives
\[
O\!\left(R\ell+\frac{2^q}{b}+H2^\ell+R\right)=O(2^q/b).
\]
Indeed, relative to $2^q/b$, the extra terms are bounded by
$b\ell/2^\ell$ and $b2^{-(1-\eta)d}$, which are uniformly bounded
because $b,\ell,d=\Theta(q)$.

Every vertex is pebbled once and every non-root value is ghosted
after use. Simultaneous actions lie in disjoint bottom subtrees,
and root combinations run separately.
\end{proof}

For parallel block cleanup, choose a cut of depth $k$, giving $m=2^k$ boundary subtrees of height $q=r-k$. We prepare it with $\Sweep$ and process the top block levelwise, since vertices on a common level have disjoint active sets. The same feasibility conditions in \cref{eq:feasible-cut} apply.

\begin{lemma}[Work and depth of block cleanup using $\Sweep$]
\label{cor:parallel-contraction}
For every cut satisfying \cref{eq:feasible-cut}, a height-$r$ subtree with no pebbles and arbitrary ghosts can be reduced to its $m$ boundary subtrees with $O(2^r)$ work, space at most $b$, and depth
\begin{equation}
O\!\left(r+\frac{2^r}{b}\right).
\label{eq:parallel-contraction}
\end{equation}
The block finishes with no pebbles, and all remaining ghosts lie in the boundary subtrees.
\end{lemma}
\begin{proof}
First prepare the $m$ boundary roots with $\Sweep$.
If $m(q+1)\le b$, give each job $\lfloor b/m\rfloor$ registers
and run them concurrently. By \cref{lem:sweep}, this takes
$O(q+m2^q/b)$ rounds.

Otherwise reserve $m-1$ registers for completed roots and give each
active job $q+1$ registers. The cut conditions imply
$b-m+1\ge\max\{(b+1)/2,q+1\}$, so a batch can contain
\[
a=\left\lfloor\frac{b-m+1}{q+1}\right\rfloor\ge\frac{b}{4(q+1)}
\]
jobs. Each batch takes $O(q+2^q/(q+1))$ rounds by \cref{lem:sweep}
and uses at most $m-1+a(q+1)\le b$ registers. The $O(m(q+1)/b)$
batches therefore take $O(m2^q/b)$ rounds, since $q(q+1)=O(2^q)$.
Both preparations use $O(m2^q)$ work within budget $b$, retaining all
boundary roots.

The $k$ levels above the boundary take $k$ rounds to pebble and
$k$ rounds to unpebble, using $2(m-1)$ work and at most $2m-1\le b$
registers. Release the boundary in one further round, adding at most
$m$ work. Substituting $q+k=r$ and $m2^q=2^r$ gives the bounds.
\end{proof}

To process smaller subtrees together, give each height-$r$ job enough space for one block to take $O(r)$ rounds. By \cref{eq:parallel-contraction}, a local budget of order $2^r/r$ suffices. We use
\begin{equation}
p(r)=\min\left\{2^r-1,\left\lceil\frac{16\cdot2^r}{r}\right\rceil\right\}.
\label{eq:resolver-space}
\end{equation}
The routine $\Resolve$ uses this nondecreasing budget $p(r)=\Theta(2^r/r)$ to process equal-height jobs in batches, repeating on the smaller subtrees left by each stage.

\begin{lemma}[Costs of parallel forest cleanup]
\label{lem:resolver}
If $b\ge p(r)$, a procedure $\Resolve(A,r,b)$ clears $A$ disjoint height-$r$ subtrees containing ghosts but no pebbles, with space at most $b$, work $O(A2^r\lambda(r))$, and depth
\begin{equation}
O\!\left(r+\frac{A2^r}{b}\lambda(r)\right).
\label{eq:forest-cleanup}
\end{equation}
\end{lemma}
\begin{proof}
Let $a$ and $t$ be the number and height of the subtrees at the current stage. Put $M=a2^t$, twice the number of their leaves. Initially $M=A2^r$. A block replaces one height-$t$ subtree by $m$ height-$q$ subtrees with $m2^q=2^t$, so $M$ is unchanged between stages.

For $t>16$, choose
\begin{equation}
q=\lceil\log t\rceil+2, \qquad m=2^{t-q}.
\label{eq:resolver-cut}
\end{equation}
Then $4t\le2^q<8t$, so $m\le2^t/(4t)$ and $q\le t/2\le2^t/(4t)$. Both cut conditions hold with the local budget $p(t)$. By \cref{cor:parallel-contraction}, one job takes $O(t)$ rounds and $O(2^t)$ work.

Run at most $\lfloor b/p(t)\rfloor$ disjoint jobs concurrently, keeping their total assigned budget within $b$. Monotonicity of $p$ ensures that at least one job fits at every stage. The $O(1+ap(t)/b)$ batches take altogether
\begin{equation}
O\!\left(\left(1+\frac{ap(t)}b\right)t\right)=O(t+M/b)
\label{eq:forest-phase}
\end{equation}
rounds and $O(M)$ work. Each batch leaves no pebbles, so its registers can be reused. After processing all $a$ jobs, replace $(a,t)$ by $(am,q)$. At $t\le16$, each subtree fits within $p(t)=2^t-1$: pebble it bottom-up and unpebble it top-down using the same batches. This also clears any leaf ghosts.

For $t>16$, the height update $f(t)=\lceil\log t\rceil+2$ satisfies $f(t)\le t/2$. Thus the heights over all stages sum to $O(r)$. Moreover, $f(f(t))\le\log t$ for $t\ge256$, so there are $O(\lambda(r))$ stages. Summing \cref{eq:forest-phase} gives $O(r+M\lambda(r)/b)$ rounds. Each stage costs $O(M)$ work, giving the work bound.
\end{proof}

We now combine these routines into our parallel pebbling algorithm. First, $\Sweep(h,s)$ computes and retains the root. Using the remaining budget $b=s-1$, we apply block cleanup to one subtree at a time. Once $p(r)\le b$, $\Resolve$ clears the remaining forest.

\begin{algorithm}[H]
\caption{Parallel spooky pebbling on complete binary trees}
\label{alg:parallel-tree}
\begin{algorithmic}[1]
\Require An empty complete binary tree $B_h$ of height $h\ge2$, and a total space budget $s$ satisfying $h+1\le s\le2^h-1$.
\Ensure All ghosts in $B_h$ are exorcised, leaving its root pebbled and every other vertex empty.
\State Run $\Sweep(h,s)$ and retain the root; set $b\gets s-1$
\State Start with the two child subtrees: $(A,r)\gets(2,h-1)$
\While{$p(r)>b$}
 \State Choose the largest $1\le k<r$ satisfying \cref{eq:feasible-cut}
 \State Clear that block in each subtree using \cref{cor:parallel-contraction}, one subtree at a time
 \State Replace the pending forest by its boundary subtrees: $(A,r)\gets(A2^k,r-k)$
\EndWhile
\State Run $\Resolve(A,r,b)$ on the pending forest
\end{algorithmic}
\end{algorithm}

\begin{theorem}[Work and depth of our parallel algorithm]
\label{thm:parallel-upper}
For $n=2^h-1$ and $h+1\le s\le n$, \cref{alg:parallel-tree} uses at most $s$ pebbles and achieves
\begin{align}
D&=O\!\left(h+\frac{nh}{s\log(s+1)}+\frac{n\lambda(h)}s\right), \label{eq:depth-upper}\\
W&=O\!\left(n\left(\frac h{\log(s+1)}+\lambda(h)\right)\right). \label{eq:depth-schedule-work}
\end{align}
\end{theorem}
\begin{proof}
Throughout cleanup, the number $A$ and height $r$ of the pending subtrees satisfy $A2^r=n+1$. In the first part, $p(r)>b$ implies both $2^r-1>b$ and $b<16\cdot2^r/r$. Hence the depth of one block, $O(r+2^r/b)$, is $O(2^r/b)$. Processing all $A$ blocks of that height takes
\[
A\,O(r+2^r/b)=O(n/b)
\]
rounds and $O(n)$ work. The largest feasible cut is used at each stage, so \cref{lem:slack} bounds their number by $O(h/\log(b+1))$.

At the transition, \cref{lem:resolver} adds $O(h+n\lambda(h)/b)$ rounds and $O(n\lambda(h))$ work. The initial root computation takes $O(h+n/s)$ rounds and $n$ work. Adding these costs and using $b=\Theta(s)$ proves the bounds. The space bound follows from the budgets in \cref{lem:sweep,cor:parallel-contraction,lem:resolver}.
\end{proof}

The two work terms reflect the cleanup stages: large subtrees contribute $O(nh/\log(s+1))$ work, and repeated batched cleanup contributes $O(n\lambda(h))$. Running more jobs concurrently leaves less space for each job and can require more cleanup stages. \Cref{sec:separation} shows that additional work is necessary for optimal depth at certain budgets.

At $s=\Theta(n)$, this extra work can be avoided by a direct block schedule. A feasible equal-depth boundary with $\Theta(s)$ roots leaves subtrees of constant height. We prepare the boundary with $\Sweep$, compute and uncompute the top block while retaining the output, and fully clean the remaining subtrees in parallel batches. The resulting schedule uses $O(n)$ work and has depth $O(h)$.

We have combined $\Sweep$, block cleanup, and $\Resolve$ into a complete parallel pebbling algorithm and established work and depth bounds for every feasible space budget. \Cref{sec:parallel-lower} proves that this algorithm achieves asymptotically optimal depth.

\section{Work lower bounds and optimality}
\label{sec:lower-bounds}

In this section, we establish work optimality for $B_h$ through a lower bound valid for every full binary tree, including unbalanced trees. We count the recomputation forced by the space budget in a sequential strategy and extend the result to the parallel model by serialization.


\subsection{Recomputation intervals}

In this section, we follow a vertex through the computation to define intervals for counting recomputation and storage. Its value may be computed, used, and released, then computed again when it is needed later. Every other vertex except the root must eventually end empty.

Now consider the last action on a non-root vertex, which must end empty. It cannot be a $\Ghost$ action, because clearing the resulting ghost would require pebbling the vertex again. Thus its final removal must be an $\Unp$ action. Write $\tau(v)$ for the step of this final removal in a sequential strategy, or its round in a parallel strategy. A parent must have both children pebbled when it is uncomputed, so their final uncomputations occur later. In parallel execution, overlapping active sets also prohibit these actions from sharing a round. Hence, for every non-root $v$ and its descendant $u$,

\begin{equation}
 \tau(u)-\tau(v)\ge\operatorname{dist}(v,u).
 \label{eq:final-removal}
\end{equation}

For a nonempty connected subtree $C$ of $T$, let $\operatorname{root}(C)$
denote the vertex of $C$ nearest the root $r$ of $T$.

We partition $C$ into smaller connected subtrees, called pieces.
For a piece $Q$, uncomputing $\parv(\operatorname{root}(Q))$
requires a pebble on $\operatorname{root}(Q)$. If $Q$ contains no pebbles, its values must
first be recomputed. Since at most $s$ pieces can contain
pebbles at once, we seek comparable sizes to limit the vertices
covered by occupied pieces and prevent rapid shrinkage under
recursive partitioning.

\begin{lemma}[Partitioning a connected region]
\label{lem:partition}
Let $C$ be a connected subtree of $T$ with $m$ vertices.
For any $1\le a<m$, we can partition $C$ into connected
pieces $Q$ satisfying $a/4<|Q|\le a$.
\end{lemma}

\begin{proof}
We repeatedly split pieces larger than $a$, keeping at least
a quarter of the vertices in each part. For a piece with
$k\ge4$ vertices, choose a centroid $x$. Removing $x$ leaves
at most three components. Each component has at most $k/2$ vertices, and their sizes sum to $k-1$.
The largest therefore has at least $(k-1)/3\ge k/4$ vertices.
Cutting its edge to $x$ gives two connected parts of sizes between $k/4$ and $3k/4$.
For $k=2$ or $3$, any edge gives these bounds.

The pieces become smaller after each split, so eventually
all have size at most $a$. Each final piece $Q$ was formed
by splitting a piece with $k>a$ vertices.
Hence $|Q|\ge k/4>a/4$.
\end{proof}

We now use these pieces to count the recomputation or cumulative
storage required between final uncomputations. The next lemma
defines the counting intervals and shows that no contribution
is counted twice.

\begin{lemma}[Recomputation and storage]
\label{lem:charging}
Fix a successful strategy and recursively partition a connected
subtree of $T$ excluding $r$. For each piece $Q$ of a subtree $C$,
write $p=\parv(\operatorname{root}(Q))$. For pieces with $p$
strictly below $\operatorname{root}(C)$, the interval $I(C,Q)$
starts immediately after the final uncomputation of
$\operatorname{root}(C)$ and ends at the final uncomputation of $p$:
\begin{equation}
 I(C,Q)=\{\tau(\operatorname{root}(C))+1,\ldots,\tau(p)\}.
 \label{eq:charge-interval}
\end{equation}
If $Q$ contains no pebbles after some step or complete round
$t$ with $\tau(\operatorname{root}(C))\le t<\tau(p)$, then
at least $|Q|$ placements occur in $Q$ during $I(C,Q)$.
Otherwise, in a parallel execution, $Q$ contributes at least
$|I(C,Q)|$ to $\Area$ during this interval.
The sets $Q\times I(C,Q)$ are pairwise disjoint across
the recursive partition.
\end{lemma}

\begin{proof}
Uncomputing $p$ requires a pebble on $\operatorname{root}(Q)$.
Suppose $Q$ contains no pebbles after time $t$ with
$\tau(\operatorname{root}(C))\le t<\tau(p)$.
Following the child dependencies within $Q$ shows that every
vertex of $Q$ must be pebbled after $t$ and before $\tau(p)$
to supply its root. Ghosts cannot supply these inputs.
Thus at least $|Q|$ placements occur in $Q\times I(C,Q)$.

Otherwise, in parallel execution,
$P_t\cap Q\ne\varnothing$ for every
$\tau(\operatorname{root}(C))\le t<\tau(p)$.
Since $P_t\subseteq U_{t+1}$,
\[
\sum_{j\in I(C,Q)}|U_j\cap Q|\ge |I(C,Q)|.
\]

Finally, consider two distinct pairs $(C,Q)$ and $(C',Q')$.
If $Q\cap Q'=\varnothing$, their vertex--time sets are disjoint.
Otherwise, exchange the pairs if necessary so that
$Q'\subseteq C'\subseteq Q$. By \cref{eq:final-removal},
\[
\max I(C,Q)
=\tau(\parv(\operatorname{root}(Q)))
<\tau(\operatorname{root}(Q))
\le\tau(\operatorname{root}(C'))
<\min I(C',Q').
\]
Their intervals are therefore disjoint. In either case,
\[
\bigl(Q\times I(C,Q)\bigr)
\cap
\bigl(Q'\times I(C',Q')\bigr)=\varnothing.
\]
\end{proof}

\subsection{Work lower bound}

In this section, we prove a lower bound on work by counting the recomputation forced by the space budget. We choose piece sizes so that a constant fraction of each subtree lies in pieces containing no pebbles when its root is finally uncomputed. The lower bound on piece size lets us repeat this argument across many partition levels.

\begin{theorem}[Work lower bound]
\label{thm:universal-work}
Every sequential strategy on a full binary tree with $n$ vertices and space at most $s$ has
\begin{equation}
\Pi=\Omega\!\left(\frac{n\log n}{\log(s+1)}\right).
\label{eq:universal-work}
\end{equation}
\end{theorem}
\begin{proof}
If $\log n/\log(s+1)$ is bounded, $\Pi\ge n$ suffices.
Otherwise, let $C_0$ be the larger child subtree of $r$.
Then $|C_0|\ge(n-1)/2$, and every vertex of $C_0$ has a final
uncomputation.

Partition each subtree $C$ of size $m>8(s+1)$ using
\cref{lem:partition} with $a=\lfloor m/(4(s+1))\rfloor$.
Each piece satisfies
\begin{equation}
 \frac{m}{32(s+1)}<|Q|\le a.
 \label{eq:work-partition}
\end{equation}
Immediately after $\operatorname{root}(C)$ is finally uncomputed, exclude
its own piece and all pieces containing pebbles.
Let $\mathcal E(C)$ contain the remaining pieces. Since at most
$s+1$ pieces are excluded,
\[
\sum_{Q\in\mathcal E(C)}|Q|
\ge m-(s+1)a\ge\frac{3m}{4}.
\]
The children of $\operatorname{root}(C)$ remain pebbled after
its uncomputation. Thus each $Q\in\mathcal E(C)$ has
$\parv(\operatorname{root}(Q))$ strictly below
$\operatorname{root}(C)$. By \cref{lem:charging}, it requires
at least $|Q|$ placements during $I(C,Q)$.

Recurse inside every piece until its size is at most $8(s+1)$.
Let $\mathcal C$ be the family of subtrees that are partitioned.
By \cref{eq:work-partition}, each partition reduces size by
a factor less than $32(s+1)$. Hence every $v\in C_0$ satisfies
\[
\bigl|\{C\in\mathcal C:v\in C\}\bigr|
=\Omega\!\left(\frac{\log n}{\log(s+1)}\right).
\]
The placements counted by \cref{lem:charging} are disjoint, so
\begin{align*}
\Pi
&\ge \sum_{C\in\mathcal C}\sum_{Q\in\mathcal E(C)}|Q|
\ge \frac34\sum_{C\in\mathcal C}|C|
 = \frac34\sum_{v\in C_0}
   \bigl|\{C\in\mathcal C:v\in C\}\bigr|
=\Omega\!\left(\frac{n\log n}{\log(s+1)}\right).
\end{align*}
\end{proof}

The serialization in \cref{sec:work-counts} preserves placements
and space. Therefore the same placement lower bound holds for
parallel strategies. Since $W\ge\Pi$, it is also a work lower
bound in either execution model.

For $B_h$, $\log n=\Theta(h)$, so this matches
\cref{thm:seq-upper} and establishes optimal work in both models.
We next use the same interval argument to bound cumulative
storage and parallel depth.

\section{Parallel depth and its work cost}
\label{sec:parallel-lower}

In this section, we prove a parallel depth lower bound for every full binary tree and show that our algorithm for $B_h$ matches it at each feasible space budget. Together with the work lower bound, this establishes optimality over tree shapes. We then show that attaining optimal depth can require asymptotically more work than minimizing work alone at the same space budget.

\subsection{Cumulative storage and depth lower bounds}
\label{sec:area}

In this section, we use the interval argument of \cref{lem:charging} to bound cumulative storage and obtain a parallel depth lower bound.

Computing the root requires at least $H(T)$ rounds.
Each placement occupies a register during its round, and at most
$s$ registers are occupied in each round. Thus
\begin{equation}
\Pi\le\Area\le sD.
\label{eq:capacity}
\end{equation}

Together with \cref{thm:universal-work}, this gives
$D=\Omega(n\log n/(s\log(s+1)))$.
To obtain a stronger bound, we should count how long pebbles remain
in memory.

\begin{theorem}[Storage and recomputation]
\label{thm:weighted}
There is a constant $c>0$ such that every parallel strategy on a full binary tree with $n$ vertices satisfies
\begin{equation}
\Pi+\frac{\Area}{L}\ge cnL, \qquad L=\lambda(n)=1+\log^*n.
\label{eq:weighted}
\end{equation}
\end{theorem}

\begin{proof}
We first prove the bound for sufficiently large $L$.
Let $C_0$ be the larger child subtree of $r$.
Its size is $m_0\ge(n-1)/2$.
Recursively partition every piece larger than $M_L=2^{L^2}$.

For a subtree $C$ of size $m>M_L$, choose
\begin{equation}
a=\left\lfloor\frac{\log m}{16L}\right\rfloor, \qquad \ell=\left\lfloor\frac12\log m\right\rfloor.
\end{equation}
By \cref{lem:partition}, each piece satisfies $a/4<|Q|\le a$.
We count pieces whose roots are at distance at least $\ell$
from $\operatorname{root}(C)$.
Fewer than $2^\ell$ vertices lie at smaller distances.
Since each piece contains at most $a$ vertices, the other pieces
cover at most
\begin{equation}
a2^\ell\le a\sqrt m\le m/4
\end{equation}
vertices. The counted pieces therefore cover at least $3m/4$.

For each counted piece $Q$, $\parv(\operatorname{root}(Q))$
lies at least $\ell-1$ edges below $\operatorname{root}(C)$.
By \cref{eq:final-removal},
\begin{equation}
|I(C,Q)|\ge\ell-1\ge La\ge L|Q|.
\end{equation}
By \cref{lem:charging}, either $Q$ requires at least $|Q|$
placements during this interval or it contributes at least
$L|Q|$ to cumulative storage.
Thus these pieces contribute at least $3|C|/4$ to $\Pi+\Area/L$.

We recurse inside every piece, including those omitted
from the current count.
For sufficiently large $L$, each piece satisfies
\begin{equation}
|Q|>a/4\ge\frac{\log m}{128L}\ge\log\log m.
\end{equation}
The last inequality follows from the monotonicity of $x/\log x$
and $\log m>L^2$, which give
\begin{equation}
\frac{\log m}{\log\log m}\ge\frac{L^2}{2\log L}\ge128L.
\end{equation}
Each partition therefore decreases the iterated logarithm
of the piece size by at most two.

Initially, $\log^*m_0=L-O(1)$.
When the recursion stops, the value is at most
$\log^*M_L=O(1+\log^*L)=o(L)$.
Every initial vertex therefore belongs to $\Omega(L)$
subtrees that are partitioned.
By \cref{lem:charging}, no placement or storage contribution
is counted twice. Hence
\begin{equation}
\Pi+\frac{\Area}{L}\ge\frac34\sum_{\text{partitioned }C}|C|=\Omega(m_0L)=\Omega(nL).
\end{equation}

For $L$ below the fixed threshold required by these estimates,
$nL=O(n)$. Thus $\Pi\ge n$ proves the bound after decreasing
$c$ if necessary.
\end{proof}

Because $\Pi\le\Area$, the theorem implies
$\Area=\Omega(n\lambda(n))$, and hence
$D=\Omega(n\lambda(n)/s)$. Combining this with the dependency
bound and \cref{thm:universal-work} gives
\begin{equation}
D=\Omega\!\left(H(T)+\frac ns\max\!\left\{\frac{\log n}{\log(s+1)},\lambda(n)\right\}\right).
\label{eq:universal-depth}
\end{equation}
For $B_h$, $\log n=\Theta(h)$ and
$\lambda(n)=\Theta(\lambda(h))$, so \cref{eq:universal-depth}
matches \cref{thm:parallel-upper}. Together with work optimality
from \cref{thm:universal-work,thm:seq-upper}, this proves \cref{thm:main}.

\subsection{Optimality over tree shapes}
\label{sec:tree-shape}

In this section, we combine the universal lower bounds with
our algorithm for $B_h$ to prove optimality over tree shapes.

\begin{proof}[Proof of \cref{thm:main-shape}]
Fix a space budget $h+1\le s\le2N-1$.
Every ordered full binary tree $T$ with $N=2^{h-1}$ leaves
has $n=2N-1=2^h-1$ vertices and height at least $h$.
Thus $\log n=\Theta(h)$ and
$\lambda(n)=\Theta(\lambda(h))$.

By \cref{thm:universal-work}, every successful sequential
or parallel strategy on $T$ with space at most $s$ has
\begin{equation}
W=\Omega\!\left(\frac{nh}{\log(s+1)}\right).
\end{equation}
By \cref{eq:universal-depth}, every such parallel strategy
also has
\begin{equation}
D=\Omega\!\left(h+\frac ns\max\!\left\{\frac{h}{\log(s+1)},\lambda(h)\right\}\right).
\end{equation}

The complete tree $B_h$ has the same number of leaves and
is feasible at this budget. By \cref{thm:main}, it attains
both bounds up to constant factors, with work and depth
minimized separately. Hence choosing another tree shape
cannot improve the asymptotic order of either quantity.
\end{proof}

Keeping the two terms separate in \cref{eq:weighted} will also let us bound the work of a depth-optimal strategy.

\subsection{Work needed for optimal depth}
\label{sec:separation}

In this section, we show that attaining optimal parallel depth
can require more work than minimizing work alone.
By \cref{eq:weighted,eq:capacity},
\begin{equation}
W\ge\Pi\ge cnL-\frac{\Area}{L}\ge cnL-\frac{sD}{L}, \qquad L=\lambda(n).
\end{equation}
At a fixed space budget, reducing depth therefore increases
this lower bound on work.

Consider $s=\Theta(n/h)$. By
\cref{eq:main-work,eq:main-depth},
\begin{equation}
\Wpar^\star(B_h,s)=\Theta(n), \qquad \Dpar^\star(B_h,s)=\Theta(h\lambda(h)).
\end{equation}
The following theorem shows that attaining this depth
requires $\Theta(n\lambda(h))$ work.

\begin{theorem}[Work at optimal depth]
\label{thm:separation}
For complete trees and budgets $s=\Theta(n/h)$, every family
of strategies with space at most $s$ and depth
$O(h\lambda(h))$ has work $\Omega(n\lambda(h))$.
\Cref{alg:parallel-tree} achieves this bound.
\end{theorem}

\begin{proof}
Set $L=\lambda(n)$. Since $\lambda(h)=\Theta(L)$,
the space and depth bounds give $sD=O(nL)$.
Thus
\begin{equation}
W\ge cnL-\frac{sD}{L}\ge cnL-O(n)=\Omega(nL)
\end{equation}
for sufficiently large $L$.
For bounded $L$, the same bound follows from $W\ge n$.

For the construction, $s=\Theta(n/h)$ gives
$\log(s+1)=\Theta(h)$.
By \cref{thm:parallel-upper}, our algorithm has work
$O(n\lambda(h))$ and depth $O(h\lambda(h))$.
It therefore attains the claimed work bound at optimal depth.
\end{proof}

At some space budgets, our algorithm achieves optimal work
and depth simultaneously. This holds when
$h/\log(s+1)=\Omega(\lambda(h))$, since the additional
$n\lambda(h)$ work does not change the asymptotic work bound.
The minimum budget $s=h+1$ satisfies this condition.
\section{Application to RNS point-addition trees}
\label{sec:cfs-case-study}

In this section, we apply block cleanup to the RNS point-addition
trees of Chevignard, Fouque, and Schrottenloher
(CFS)~\cite{CFS2026} to reduce repeated computations under
the same workspace budget.

For a fixed RNS prime $p$, each intermediate value is a coordinate
triple $P=(X,Y,Z)\in\mathbb F_p^3$. Leaves load precomputed triples
selected by window indices. At a non-root internal vertex,
the circuit combines its two child triples using the projective
point-addition formulas evaluated modulo $p$. Denoting this binary
operation by $\operatorname{Add}_p$, the circuit performs
\begin{equation}
\ket{P_1}\ket{P_2}\ket{0}
\longmapsto
\ket{P_1}\ket{P_2}\ket{\operatorname{Add}_p(P_1,P_2)}.
\label{eq:cfs-add}
\end{equation}
Each stored triple corresponds to a pebble, and uncomputing it
requires both child triples. The root computation updates the
output accumulators, after which all temporary values and pending
measurement phases must be cleared. These dependencies and cleanup
requirements give a spooky pebbling problem.

We adapt block cleanup to the public circuits, choosing boundaries
within the recursive scheme according to their circuit costs
and workspace requirements. At the top of a cleanup block,
the public phase-correction circuit replaces a pebble--unpebble
pair when both child values are available.

We compare block cleanup with the public CFS schedule using the
same tree, input order, arithmetic circuits, root operation,
and workspace budget.
\Cref{tab:cfs-calls} shows the circuit-call counts for P-256
at RNS prime $521177$. The savings come from fewer leaf lookups
and point-addition computations, while the required phase
corrections, internal uncomputations, and root operation
remain unchanged.

\begin{table}[H]
\centering
\small
\begin{tabular}{@{}lrr@{}}
\toprule
Operation & CFS schedule & Block cleanup \\
\midrule
Leaf lookup                         & 82 & 57 \\
Leaf uncomputation                  & 82 & 57 \\
$\operatorname{Add}_p$ computation   & 46 & 35 \\
$\operatorname{Add}_p$ uncomputation &  7 &  7 \\
Direct phase correction             & 10 & 10 \\
Root operation                     &  1 &  1 \\
\bottomrule
\end{tabular}
\caption{Circuit-call counts for P-256 at RNS prime $521177$.}
\label{tab:cfs-calls}
\end{table}

To quantify the circuit savings, we weight each call by its
Toffoli/AND gate count under the public CFS cost model, counting
each Toffoli gate and temporary logical-AND computation as one
unit.\footnote{
The public model assigns zero gate cost to inverse dummy lookups.
The cryptographic-size estimates use the public implementation's
symbolic lookup and arithmetic shortcuts.}
\Cref{tab:cfs-application} reports the total per RNS prime.
Block cleanup reduces this cost by $20.86\%$ for P-224 and
approximately $22.6\%$ for P-256, with unchanged peak qubits.

\begin{table}[H]
\centering
\small
\setlength{\tabcolsep}{3.2pt}
\renewcommand{\arraystretch}{1.15}
\begin{tabular}{@{}lrrrrrc@{}}
\toprule
Curve & Leaves & RNS prime & CFS schedule & Block cleanup
& Reduction & \shortstack{Qubits\\Peak / cap} \\
\midrule
P-224 & 17 & 262139 & 9,306,464 & 7,365,187
& 20.86\% & 1165 / 1165 \\
P-256 & 19 & 521177 & 11,559,064 & 8,939,833
& 22.66\% & 1303 / 1325 \\
P-256 & 19 & 499979 & 11,640,184 & 9,007,225
& 22.62\% & 1303 / 1325 \\
\bottomrule
\end{tabular}
\caption{Toffoli/AND counts per RNS prime. Both schedules have
the reported Qarton peak qubits and satisfy the same conservative
workspace cap including circuit scratch.}
\label{tab:cfs-application}
\end{table}

We checked output values, relative phases, and complete cleanup
using basis and superposition simulations.
Thus fewer recomputations yield lower circuit costs under
the same workspace and cleanup requirements.
Implementation and reproduction details will be provided at
\url{https://github.com/red1108/pebble}.

\section{Discussion}

When space is limited, cleanup is the work bottleneck on complete binary trees. Our key idea is to clean up the connected groups of vertices together to reduce repeated computation. It gives asymptotically optimal work at every feasible space budget. For associative operations, we also proved that a complete balanced tree is asymptotically optimal for work and parallel depth among all full binary trees.

For parallel spooky pebbling, we divide the space budget among
independent jobs and run them concurrently. This gives
asymptotically optimal depth, but each job has less space for
cleanup and may require more recomputation. We also prove
that this extra work is unavoidable at some space budgets.

During the computation, if sharing intermediate values is allowed,
the dependencies can form a general DAG. Our lower-bound argument
relies on each non-root vertex having a single parent, so it does
not directly extend to this setting. Extending our algorithms and
establishing tight work and depth bounds for these graphs remain open.
For complete binary trees, the full work--depth tradeoff at
a fixed space budget also remains open.

\section*{Acknowledgements}
This work was supported by the National Research Foundation of Korea (NRF) through a grant funded by the Ministry of Science and ICT (Grant No. RS-2025-00515537), the Institute for Information and Communications Technology Promotion (IITP) grant funded by the Korean government (MSIP) (Grant Nos. RS-2019-II190003 and RS-2025-02304540), the National Research Council of Science and Technology (NST) (Grant No. GTL25012-401), and the Korea Institute of Science and Technology Information (Grant No. P26028).

\section*{AI disclosure}

The lower-bound proofs were developed through iterative
exchanges with GPT-6 Astra and GPT-6.1 Sol. We also used these models to polish the text, search for applications,
and assist with the code and computational experiments.
The authors are responsible for the proofs, code, and
results presented in this paper.

\printbibliography

@article{Bennett1989,
  author={Bennett, Charles H.},
  title={Time/Space Trade-Offs for Reversible Computation},
  journal={SIAM Journal on Computing},
  volume={18}, number={4}, pages={766--776}, year={1989},
  doi={10.1137/0218053}
}

@inproceedings{Cook1973,
  author={Cook, Stephen A.},
  title={An Observation on Time-Storage Trade Off},
  booktitle={Proceedings of the Fifth Annual ACM Symposium on Theory of Computing},
  pages={29--33}, year={1973}, publisher={ACM}
}

@online{Gidney2019,
  author={Gidney, Craig},
  title={Spooky Pebble Games and Irreversible Uncomputation},
  year={2019},
  url={https://algassert.com/post/1905},
  note={Online technical article}
}

@article{KSS2025,
  author={Kornerup, Niels and Sadun, Jonathan and Soloveichik, David},
  title={Tight Bounds on the Spooky Pebble Game: Recycling Qubits with Measurements},
  journal={Quantum}, volume={9}, pages={1636}, year={2025},
  doi={10.22331/q-2025-02-18-1636},
  eprint={2110.08973}, eprinttype={arxiv}, eprintclass={quant-ph}
}

@article{QL2025,
  author={Quist, Arend-Jan and Laarman, Alfons},
  title={Trade-offs between Classical and Quantum Space Using Spooky Pebbling},
  journal={Logical Methods in Computer Science},
  volume={21}, number={4}, pages={29:1--29:26}, year={2025},
  doi={10.46298/LMCS-21(4:29)2025},
  eprint={2401.10579}, eprinttype={arxiv}
}

@inproceedings{BHL2022,
  author={Blocki, Jeremiah and Holman, Blake and Lee, Seunghoon},
  title={The Parallel Reversible Pebbling Game: Analyzing the Post-Quantum Security of {iMHFs}},
  booktitle={Theory of Cryptography---TCC 2022},
  pages={52--79}, year={2022}, publisher={Springer}
}

@inproceedings{KMRV2026,
  shorthand={KMRV26},
  author={Kahanamoku-Meyer, Gregory D. and Ragavan, Seyoon and Van Kirk, Katherine},
  title={Parallel Spooky Pebbling Makes {Regev} Factoring More Practical},
  booktitle={Advances in Cryptology---EUROCRYPT 2026},
  year={2026}, publisher={Springer},
  doi={10.1007/978-3-032-25291-3_14},
  eprint={2510.08432}, eprinttype={arxiv},
  note={We use arXiv version 3, revised July 3, 2026}
}

@misc{CFS2026,
  author={Chevignard, Cl{\'e}mence and Fouque, Pierre-Alain and Schrottenloher, Andr{\'e}},
  title={Reducing the Number of Qubits in Quantum Discrete Logarithms on Elliptic Curves},
  howpublished={Cryptology {ePrint} Archive, Paper 2026/280},
  year={2026},
  eprint={2026/280}, eprinttype={IACR},
  url={https://eprint.iacr.org/2026/280},
  note={Minor revision of the EUROCRYPT 2026 publication}
}
\end{document}